\documentclass[runningheads]{llncs}
\usepackage[T1]{fontenc}
\usepackage{graphicx}
\usepackage{amsmath,amssymb}
\usepackage[bookmarks=false,colorlinks=true,citecolor=blue]{hyperref}
\usepackage{comment}
\usepackage[linesnumbered,ruled,vlined,boxed,commentsnumbered,linesnumbered, longend]{algorithm2e}

\usepackage{etoolbox}
\AtEndEnvironment{proof}{\qed}
\newtheorem{observation}[theorem]{Observation}{\bfseries}{\itshape}
\begin{document}
\title{Parameterized algorithms for $k$-Inversion\thanks{Supported by ANRF research grants ANRF/ECRG/2024/001049/ENS, CRG/2022/006770, and MTR/2022/000692.}}
%
%
%
%
\author{
Dhanyamol Antony\inst{1} \and
L.~Sunil Chandran\inst{2} \and
Dalu Jacob\inst{3} \and
R.~B.~Sandeep\inst{4}
}

\authorrunning{D. Antony et al.}

\institute{
School of Data Science, Indian Institute of Science Education and Research Thiruvananthapuram,\\
Thiruvananthapuram, India\\
\email{dhanyamolantony@iisertvm.ac.in}
\and
Department of Computer Science and Automation, Indian Institute of Science,\\
Bengaluru, India\\
\email{sunil@iisc.ac.in}
\and
Department of Mathematics, Indian Institute of Technology Delhi,\\
New Delhi, India\\
\email{dalujacob@iitd.ac.in}
\and
Department of Computer Science and Engineering, Indian Institute of Technology Dharwad,\\
Dharwad, India\\
\email{sandeeprb@iitdh.ac.in}
}

\maketitle              
\begin{abstract}
Inversion of a directed graph $D$ with respect to a vertex subset $Y$ is the directed graph obtained from $D$ 
by reversing the direction of every arc whose endpoints both lie in $Y$. More generally, the inversion of $D$ 
with respect to a tuple $(Y_1, Y_2, \ldots, Y_\ell)$ of vertex subsets is defined as the directed graph obtained 
by successively applying inversions with respect to $Y_1, Y_2, \ldots, Y_\ell$. 
Such a tuple is called a \emph{decycling family} of $D$ if the resulting graph is acyclic.

In the \textsc{$k$-Inversion} problem, the input consists of a directed graph $D$ and an integer $k$, and the task is to decide whether $D$ admits a decycling family of size at most $k$. Alon et al.\ (SIAM J.\ Discrete Math., 2024) proved that the problem is NP-complete for every fixed value of $k$, thereby ruling out XP algorithms, and presented a fixed-parameter tractable (FPT) algorithm parameterized by $k$ for tournament inputs.

In this paper, we generalize their algorithm to a broader variant of the problem on tournaments and subsequently use this result to obtain an FPT algorithm for \textsc{$k$-Inversion} when the underlying undirected graph of the input is a block graph. Furthermore, 
we obtain an algorithm for \textsc{$k$-Inversion} on general directed graphs with running time $2^{O(\mathrm{tw}(k + \mathrm{tw}))} \cdot n^{O(1)}$, where $\mathrm{tw}$ denotes the treewidth of the underlying graph.

\end{abstract}
\newpage
\section{Introduction}

Inversion of a directed graph $D$ with respect to a vertex subset $Y$, denoted by
$D \oplus Y$, is the directed graph obtained from $D$ by reversing the direction
of every arc whose endpoints both lie in $Y$. More generally, the inversion of
$D$ with respect to a tuple $\mathcal{Y} = (Y_1, Y_2, \ldots, Y_\ell)$ of vertex
subsets, denoted by $D \oplus \mathcal{Y}$, is 
obtained by successively applying inversions with respect to
$Y_1, Y_2, \ldots, Y_\ell$. Such a tuple is called a \emph{decycling family} of $D$
if the resulting graph is acyclic. The minimum cardinality of a
decycling family of $D$ is known as the \emph{inversion number} of $D$ and is
denoted by $\mathrm{inv}(D)$.

Note that the order in which the inversions are applied has no effect on the
final graph. Indeed, an arc $uv$ is present in $D \oplus \mathcal{Y}$ if and only
if either $uv$ is an arc in $D$ and the number of sets in $\mathcal{Y}$ that
contain $\{u,v\}$ is even, or $vu$ is an arc in $D$ and the number of sets in
$\mathcal{Y}$ that contain $\{u,v\}$ is odd. 

Belkhechine et al.~\cite{belkhechine2010inversion} initiated the study of inversion
in 2010 and investigated the inversion number in connection with the Boolean
dimension of graphs. Bang-Jensen, da Silva, and Havet~\cite{bang2022inversion}
studied the inversion number in relation to the cycle transversal number, the
cycle arc-transversal number, and the cycle packing number. They conjectured that
for any two oriented graphs $L$ and $R$,
$\mathrm{inv}(L \rightarrow R) = \mathrm{inv}(L) + \mathrm{inv}(R)$, where
$L \rightarrow R$ denotes the \emph{dijoin} of $L$ and $R$, that is, the oriented
graph obtained from $L$ and $R$ by adding all possible arcs from $L$ to $R$.
This conjecture, known as the \emph{dijoin conjecture}, led to a series of
subsequent works~\cite{behague2025note,wang2024inversion,alon2024invertibility,behague2025case,aubian2025problems}.
Recently, Alon et al.~\cite{alon2024invertibility} and Aubian et
al.~\cite{aubian2025problems} disproved the conjecture. The former also showed
that $\mathrm{inv}(D) \leq (1+o(1))n$ for every directed graph $D$ on $n$ vertices.

Havet, Hörsch, and Rambaud~\cite{havet2024diameter} studied the diameter of a graph
whose vertices correspond to oriented graphs, with an arc from one vertex to
another if the second oriented graph can be obtained from the first by a single
inversion. They established connections between this diameter and several graph
parameters, including the star chromatic number, acyclic chromatic number,
oriented chromatic number, and treewidth. The connection with treewidth was
further investigated by Wang et al.~\cite{wang2024inversion}.

Yuster~\cite{yuster2025tournament} obtained bounds on the inversion number of
tournaments. There are also recent works that consider inversion operations with
respect to different target properties; see, for example,
Duron et al.~\cite{duron2025minimum} and Belkhechine and Ben
Salha~\cite{belkhechine2021making}.


\medskip
\noindent
In the \textsc{$k$-Inversion} problem, the input consists of a directed graph $D$
and an integer $k$, and the objective is to decide whether $\mathrm{inv}(D) \leq
k$.

\medskip
\noindent
\textbf{Known algorithmic results.}
Algorithmic results related to inversion are relatively sparse in
comparison with the extensive non-algorithmic literature. Bang-Jensen, da Silva,
and Havet~\cite{bang2022inversion} proved that \textsc{$k$-Inversion} is
NP-complete for $k=1$, ruling out XP algorithms. Alon et al.~\cite{alon2024invertibility} extended this
result and established NP-completeness for every fixed value of $k$. 
They also showed that the problem is fixed-parameter
tractable when parameterized by $k$ on tournament inputs, using an
iterative-compression-based algorithm. Very recently, Bang-Jensen et
al.~\cite{bang2025making} conducted an algorithmic study of a variant of
\textsc{$k$-Inversion} in which each set used for inversion is required to have
bounded size.
To the best of our knowledge, no other algorithmic results are known for the problem.
We address this gap by presenting three fixed-parameter tractable (FPT) algorithms.

\medskip
\noindent
\textbf{Our results.}
Our first contribution is an FPT algorithm for a generalization of
\textsc{$k$-Inversion} on tournaments.

A decycling family $\mathcal{Y} = (Y_1, Y_2, \ldots, Y_k)$ of a directed graph $D$
induces a \emph{characteristic vector} $\mathbf{u}_i \in \{0,1\}^k$ for each
vertex $u_i$ of $D$. The $j$th coordinate of $\mathbf{u}_i$ is $1$ if
$u_i \in Y_j$ and $0$ otherwise. The \emph{weight} of $\mathbf{u}_i$ with respect
to $\mathcal{Y}$, denoted by $\mathrm{wt}_{\mathcal{Y}}(\mathbf{u}_i)$, is the
number of $1$s in $\mathbf{u}_i$. We omit the subscript when the decycling family
is clear from the context.

We define the \textsc{Weight-Restricted $k$-Inversion on Tournaments} problem as
follows. The input consists of a tournament $T$, an integer $k$, and, for each
vertex $u_i$ of $T$, a set $A_i$ of allowed weights. The objective is to decide
whether there exists a decycling family
$\mathcal{Y}$ of size $k$ such that
$\mathrm{wt}(\mathbf{u}_i) \in A_i$ for every vertex $u_i$ of $T$.

\begin{theorem}
    \label{thm:wrkit}
    \textsc{Weight-Restricted $k$-Inversion on Tournaments} admits an FPT
    algorithm when parameterized by $k$.
\end{theorem}

The algorithm adapts the iterative-compression-based algorithm of
Alon et al.~\cite{alon2024invertibility} for \textsc{$k$-Inversion on
Tournaments}.

Next, we consider directed graphs whose underlying undirected graph is a block
graph. Using Theorem~\ref{thm:wrkit} as a subroutine, we obtain an FPT algorithm
for \textsc{$k$-Inversion} on this class by dynamic programming over the block
tree of the underlying graph.

\begin{theorem}
    \label{thm:block}
    \textsc{$k$-Inversion}, parameterized by $k$, admits an FPT algorithm when the underlying undirected
    graph of the input digraph is a block graph.
\end{theorem}

The \textsc{$k$-Inversion} problem can be expressed by an MSO$_2$ formula whose
length depends only on $k$. By Courcelle's
theorem~\cite{courcelle1990monadic}, this implies that \textsc{$k$-Inversion}, parameterized by $k+\mathrm{tw}$,
admits an FPT algorithm. 
Here $\mathrm{tw}$
denotes the treewidth of the underlying undirected graph. Our final result
provides an explicit and more efficient algorithm for this parameterization, which is a dynamic programming over a tree decomposition.

\begin{theorem}
    \label{thm:tw}
    \textsc{$k$-Inversion} admits an algorithm with running time
    $2^{O(\mathrm{tw} (k + \mathrm{tw}))} \cdot n^{O(1)}$, where $\mathrm{tw}$ is the treewidth of
    the underlying undirected graph of the input digraph.
\end{theorem}

Moreover, with standard techniques, our algorithms can be made constructive and output a decycling family for yes-instances.

Note that it suffices to prove our results for oriented graphs. Indeed, if a
directed graph $D$ contains a self-loop, then no decycling family exists for
$D$. Similarly, if $D$ contains a pair of opposite arcs $uv$ and $vu$, then $D$
admits no decycling family. Moreover, if there are multiple parallel arcs from
$u$ to $v$ in $D$, then removing all but one such arc yields an equivalent
instance. Therefore, throughout the remainder of the paper, we restrict our
attention to oriented input graphs.

We prove Theorem~\ref{thm:wrkit} in Section~\ref{sec:wkit}, Theorem~\ref{thm:block} 
in Section~\ref{sec:block}, and Theorem~\ref{thm:tw} in Section~\ref{sec:tw}. 
We follow standard graph-theoretic and set-theoretic notations.
For background on the parameterized algorithmic techniques and related concepts used in this paper, we refer the reader to the book by Cygan et al.~\cite{cygan2015parameterized}.
The number of vertices in the input graph is always denoted by $n$.
By $[t]$, we denote the set $\{1,2,\ldots,t\}$.
For technical reasons, we allow a decycling family of a subgraph of a directed graph to include vertices outside the subgraph; however, such vertices do not induce any arc changes within the subgraph.
Notation and terminology not introduced in this section are defined just before their first use.

\section{\textsc{Weight-Restricted $k$-Inversion on Tournaments}}
\label{sec:wkit}

Recall that in \textsc{Weight-Restricted $k$-Inversion on Tournaments} (WKIT),
we are given a tournament $T_0$ on vertex set
$\{u_1,u_2,\ldots,u_n\}$, an integer $k$, and a tuple $\mathcal{A} = (A_1, A_2, \ldots, A_n)$ with
$A_i\subseteq\{0,1,\ldots,k\}$. The task is to decide whether there exists a
decycling family $\mathcal{Y}$ of size $k$ such that
$\mathrm{wt}(\mathbf{u}_i)\in A_i$ for every $i\in[n]$. We call these
\textit{weight constraints}.

We give an FPT algorithm for the problem, parameterized by $k$, by adapting the
iterative-compression algorithm of Alon et al.~\cite{alon2024invertibility} for
\textsc{$k$-Inversion on Tournaments}.

\begin{proposition}[\cite{alon2024invertibility}]
    \label{pro:kit}
    \textsc{$k$-Inversion on Tournaments}, parameterized by $k$, admits an FPT algorithm.
    Moreover, given a tournament $T_0$ and an integer $k$, the algorithm
    returns a decycling family of size $k$, if one exists.
\end{proposition}

We first define the following auxiliary \emph{compression} version of the
problem.

\medskip
\noindent
\textsc{Compression Weight-Restricted $k$-Inversion on Tournaments} (CWKIT).
An instance of the problem consists of:
\begin{itemize}
    \item a tournament $T_0$ on $n$ vertices $u_1, u_2, \ldots, u_n$,
    \item an integer $k$ (the parameter),
    \item a tuple $\mathcal{A} = (A_1, A_2, \ldots, A_n)$ such that
    $A_i\subseteq\{0,1,\ldots,k\}$ for all $i\in[n]$, and
    \item a decycling family $\mathcal{X}$ of size $s$ for $T_0$, where
    $s=f(k)$ for some computable function $f$.
\end{itemize}
The objective is to decide whether there exists a decycling family
$\mathcal{Y}$ of size $k$ for $T_0$ such that
$\mathrm{wt}(\mathbf{u}_i)\in A_i$ for every $i\in[n]$.

\medskip
We now explain why it suffices to solve the compression version in FPT time.

\begin{lemma}
    \label{lem:iterative-compression}
    If \textsc{CWKIT} admits an FPT algorithm running in time
    $g(k)\cdot n^{O(1)}$, then \textsc{WKIT} can be solved in time
    $g'(k)\cdot n^{O(1)}$ for some computable function $g'$.
\end{lemma}

\begin{proof}
Let $T_0$ be the input tournament.
Using Proposition~\ref{pro:kit}, we decide in FPT time whether $T_0$ admits a decycling family of size $k$.
If no such family exists, then clearly $(T_0, k, \mathcal{A})$ is a no-instance.

Suppose that $T_0$ admits a decycling family $\mathcal{Z}$ as obtained by Proposition~\ref{pro:kit}. 
Then $(T_0,k,\mathcal{A},\mathcal{Z})$ is a valid instance of
\textsc{CWKIT}. Now, running the assumed FPT algorithm for \textsc{CWKIT}
decides whether there exists a decycling family
$\mathcal{Y}$ of size $k$ for $T_0$ satisfying all weight constraints.
\end{proof}

Hence, from now on we focus on the compression version \textsc{CWKIT}.


Fix an instance $(T_0,k,\mathcal{A},\mathcal{X})$ of \textsc{CWKIT}. 
Assume that $|\mathcal{X}| = s$.
Let $T = T_0\oplus \mathcal{X}$. 
We assume, without loss of generality, that $u_1,u_2,\ldots,u_n$ is a topological ordering of $T$.
Further assume that $A_i$ is nonempty for each $i\in [n]$. Indeed, if $A_i = \emptyset$ for any $i\in [n]$,
then the instance is a no-instance.

Let $\mathcal{Y}=(Y_1,\ldots,Y_k)$ be any decycling family for $T_0$ and consider the concatenated sequence
\[
\mathcal{W}:=(X_1,\ldots,X_s,\ Y_1,\ldots,Y_k).
\]
Applying $\mathcal{W}$ to $T$ yields the same transitive tournament as applying $(Y_1,\ldots,Y_k)$ to $T_0$.
This is because $T\oplus \mathcal{W} = (T\oplus \mathcal{X})\oplus \mathcal{Y} = T_0\oplus \mathcal{Y}$,
where the second equality follows from the fact that if $T_0\oplus \mathcal{X} = T$, then $T\oplus \mathcal{X} = T_0$.
It is convenient to reason about $\mathcal{W}$ rather than $\mathcal{Y}$ directly.

For a vertex $u_i$, let $\mathbf{u}_i\in\{0,1\}^{s+k}$ be its characteristic vector with respect to $\mathcal{W}$.

Let
\[
J := \{0,1\}^{s+k},
\qquad |J|=2^{s+k},
\]
whose elements we interpret as all possible $(s+k)$-length characteristic vectors.

For $1\leq i\leq n$, define $J_i$ as the set of $(s+k)$-length characteristic vectors $x\in J$ respecting $\mathcal{X}$ (i.e., $j$\textsuperscript{th} element of $x$ is 1 if and only if $u_i$ appears in $X_j$) such that
$\mathrm{wt}(\mathbf{x}^{\mathrm{last}(k)})\in A_i$, where $\mathbf{x}^{\mathrm{last}(k)}$ 
denotes the vector comprising the last $k$ elements in $\mathbf{x}$.
Clearly, only vectors from $J_i$ are potential candidates to be assigned to the vertex $u_i$.
Since we assumed that all $A_i$ are nonempty, it follows that all $J_i$ are nonempty.

For $1\leq t\leq n$, and for any assignment $\mathbf{u}=(\mathbf{u}_1,\ldots,\mathbf{u}_t)$ with $\mathbf{u}_i\in J_i$, define
\[
B(\mathbf{u}) := \{(\mathbf{u}_a,\mathbf{u}_b,\mathbf{u}_c)\mid 1\le a<b<c\le t\},
\]
the set of all (ordered-by-position) triples of characteristic vectors appearing along the order $u_1,\ldots,u_t$.

We say that a triple $(\mathbf{a},\mathbf{b},\mathbf{c})\in J^3$ is \emph{bad} if
\[
\mathbf{a}\cdot \mathbf{b} = \mathbf{b}\cdot \mathbf{c}
\quad\text{and}\quad
\mathbf{a}\cdot \mathbf{b} \neq \mathbf{a}\cdot \mathbf{c}.
\]

The following lemma is observed in~\cite{alon2024invertibility}.

\begin{lemma}\label{lem:bad-triples-transitivity}
Let $T$ be a transitive tournament with topological order $u_1,\ldots,u_n$.
Fix any family $\mathcal{W}$ and let $\mathbf{u}_i\in\{0,1\}^{|\mathcal{W}|}$ be the corresponding characteristic vectors.
Let $T_{\mathcal{W}} = T\oplus \mathcal{W}$.
Then $T_{\mathcal{W}}$ is transitive if and only if $B(\mathbf{u})$ contains no bad triple.
\end{lemma}

\begin{proof}
We use the standard characterization of transitive tournaments: a tournament is transitive if and only if it has no directed $3$-cycle.

Fix indices $a<b<c$. In the transitive tournament $T$, the arcs among $\{u_a,u_b,u_c\}$ are
$u_a\to u_b$, $u_a\to u_c$, and $u_b\to u_c$.
The triple $(u_a,u_b,u_c)$ forms a directed $3$-cycle in $T_{\mathcal{W}}$ precisely when 
either $u_au_b$ and $u_bu_c$ are flipped (and $u_au_c$ is not flipped), or when $u_au_c$ is flipped
(and $u_au_b$ and $u_bu_c$ are not flipped).
This happens exactly when the above conditions are satisfied.
Hence, $T_{\mathcal{W}}$ contains a directed triangle if and only if there exists $a<b<c$ such that
$(\mathbf{u}_a,\mathbf{u}_b,\mathbf{u}_c)$ is a bad triple.
Equivalently, $T_{\mathcal{W}}$ is transitive if and only if no bad triple appears in $B(\mathbf{u})$.
\end{proof}

For $t\geq 3$, let
\[
\mathcal{B}_t
=
\big\{
B(\mathbf{u}) \mid
\mathbf{u} = (\mathbf{u}_1, \mathbf{u}_2, \ldots, \mathbf{u}_t),
\ \mathbf{u}_i\in J_i \text{ for } i\in [t]
\big\}.
\]


Computing $\mathcal{B}_t$ directly from the definition is costly, as there can be up to
$2^{s+k}$ choices for each $\mathbf{u}_i$.
However, as observed in~\cite{alon2024invertibility}, there are only at most
$2^{2^{3(s+k)}}$ distinct sets in $\mathcal{B}_t$, even when $t = n$.
Moreover, these sets can be computed much more efficiently as follows.

For $t\geq 4$, given a set $B'\in\mathcal{B}_{t-1}$ and a choice
$\mathbf{x}\in J_t$ for $\mathbf{u}_t$, define
\begin{equation}\label{eq:transition}
S(B',\mathbf{x})
\;:=\;
B'\ \cup\ \bigcup_{(\mathbf{v}_1,\mathbf{v}_2,\mathbf{v}_3)\in B'}
\Big\{(\mathbf{v}_1,\mathbf{v}_2,\mathbf{x}),\,
      (\mathbf{v}_1,\mathbf{v}_3,\mathbf{x}),\,
      (\mathbf{v}_2,\mathbf{v}_3,\mathbf{x})\Big\}.
\end{equation}
Intuitively, when appending a new position $t$ carrying vector $\mathbf{x}$,
every previous triple generates three new triples involving $\mathbf{x}$.

\begin{lemma}\label{lem:cwkit:transition-correct}
For any $t\ge 4$, any $\mathbf{u}'\in J^{t-1}$, and any $\mathbf{x}\in J_t$,
let $\mathbf{u}=(\mathbf{u}',\mathbf{x})$. Then
\[
B(\mathbf{u}) = S\big(B(\mathbf{u}'),\mathbf{x}\big).
\]
\end{lemma}

\begin{proof}
Let $\mathbf{u}' = (\mathbf{u}_1, \mathbf{u}_2, \ldots, \mathbf{u}_{t-1})$. Then $\mathbf{u} = (\mathbf{u}_1, \mathbf{u}_2, \ldots, \mathbf{u}_{t-1}, \mathbf{u}_t = \mathbf{x})$.
Let $B' = B(\mathbf{u}')$.
Let $(\mathbf{u}_a, \mathbf{u}_b, \mathbf{u}_c)$ be any triple in $B(\mathbf{u})$.
Note that $1\leq a < b < c \leq t$ and
$\mathbf{u}_a\in J_a$, $\mathbf{u}_b\in J_b$, and $\mathbf{u}_c\in J_c$.

If $c \leq t-1$, then
$(\mathbf{u}_a, \mathbf{u}_b, \mathbf{u}_c)\in B'\subseteq S(B', \mathbf{x})$.
Hence, assume that $c = t$, that is, $\mathbf{u}_c = \mathbf{x}$.
Then $1\leq a < b \leq t-1$.
Since $t-1\geq 3$ and each $J_i$ is nonempty, there exists an index
$\ell \in [t-1]\setminus\{a,b\}$ such that either $\ell < a$, or $a < \ell < b$,
or $b < \ell$.

Assume that $\ell < a$.
Then, by definition, $B'$ contains the triple
$(\mathbf{u}_\ell, \mathbf{u}_a, \mathbf{u}_b)$.
Consequently, $S(B', \mathbf{x})$ contains the triple
$(\mathbf{u}_a, \mathbf{u}_b, \mathbf{x})$, as required.
The other two cases can be proved analogously.

For the reverse direction, assume that $S(B', \mathbf{x})$ contains a triple
$(\mathbf{u}_a, \mathbf{u}_b, \mathbf{u}_c)$.
If this triple belongs to $B'$, then, since $B'\subseteq B(\mathbf{u})$,
it is contained in $B(\mathbf{u})$.
Otherwise, the triple was added to $S(B', \mathbf{x})$ due to some triple in $B'$
in which both $\mathbf{u}_a$ and $\mathbf{u}_b$ appear, with $\mathbf{u}_a$
preceding $\mathbf{u}_b$.
In this case, $\mathbf{u}_c=\mathbf{x}$ and
$(\mathbf{u}_a, \mathbf{u}_b, \mathbf{x})$ belongs to $B(\mathbf{u})$ by definition.
\end{proof}

Now, for $t\geq 4$, we compute $\mathcal{B}_t$ as
\[
\bigcup_{\mathbf{x}\in J_t,\; B'\in \mathcal{B}_{t-1}} S(B', \mathbf{x}).
\]
It remains to prove that the definition of $\mathcal{B}_t$ matches its computation.

\begin{lemma}
    \label{lem:cwkit:bt}
    For $n\geq t\geq 4$,
    \[
    \mathcal{B}_t = \bigcup_{\mathbf{x}\in J_t,\; B'\in \mathcal{B}_{t-1}} S(B', \mathbf{x}).
    \]
\end{lemma}

\begin{proof}
    Let $\mathcal{D} = \bigcup_{\mathbf{x}\in J_t,\; B'\in \mathcal{B}_{t-1}} S(B', \mathbf{x})$.
    Let $B\in \mathcal{B}_t$.
    By the definition of $\mathcal{B}_t$, there exists
    $\mathbf{u} = (\mathbf{u}_1, \mathbf{u}_2, \ldots, \mathbf{u}_t)$
    such that $B = B(\mathbf{u})$.

    Let $\mathbf{u}' = (\mathbf{u}_1, \mathbf{u}_2,\ldots, \mathbf{u}_{t-1})$
    and let $B' = B(\mathbf{u}')$.
    By Lemma~\ref{lem:cwkit:transition-correct},
    $B = S(B', \mathbf{x})$.
    Therefore, $B\in \mathcal{D}$.

    For the reverse direction, let $B\in \mathcal{D}$.
    Then $B = S(B', \mathbf{x})$ for some
    $B'\in \mathcal{B}_{t-1}$ and $\mathbf{x}\in J_t$.
    By the definition of $\mathcal{B}_{t-1}$, there exists
    $\mathbf{u}' = (\mathbf{u}_1, \mathbf{u}_2, \ldots, \mathbf{u}_{t-1})$
    such that $\mathbf{u}_i\in J_i$ for all $i\in [t-1]$.
    Let
    \[
    \mathbf{u} = (\mathbf{u}_1, \mathbf{u}_2, \ldots, \mathbf{u}_{t-1}, \mathbf{u}_t = \mathbf{x}).
    \]
    By Lemma~\ref{lem:cwkit:transition-correct},
    $B(\mathbf{u}) = S(B', \mathbf{x}) = B$.
    Therefore, $B\in \mathcal{B}_t$.
\end{proof}

\begin{lemma}
    \label{lem:cwkit-yes-instance}
    Let $I = (T_0, k, \mathcal{A}, \mathcal{X})$ be an instance of \textsc{CWKIT}.
    Assume that $T_0$ has at least $3$ vertices.
    Then $I$ is a yes-instance if and only if there exists a set
    $B\in \mathcal{B}_n$ such that $B$ contains no bad triples.
\end{lemma}

\begin{proof}
    Assume that $I$ is a yes-instance.
    Then there exists a decycling family
    $\mathcal{Y} = (Y_1, Y_2, \ldots, Y_k)$
    satisfying the weight constraints.
    Let
    $\mathbf{u} = (\mathbf{u}_1, \mathbf{u}_2, \ldots, \mathbf{u}_n)$
    be the vector of characteristic vectors of the vertices of $T_0$
    corresponding to
    $\mathcal{W} = (X_1, X_2, \ldots, X_s, Y_1, Y_2, \ldots, Y_k)$.
    By definition, $B(\mathbf{u})\in \mathcal{B}_n$ and by Lemma~\ref{lem:bad-triples-transitivity},
    $B(\mathbf{u})$ contains no bad triples.

    For the converse, assume that there exists a set
    $B\in \mathcal{B}_n$ such that $B$ contains no bad triples.
    We prove that $I$ is a yes-instance by induction on $n$.

    The base case is $n = 3$.
    Recall that
    \[
    \mathcal{B}_3 =
    \Big\{\, \{(\mathbf{a},\mathbf{b},\mathbf{c})\}
    \;\Big|\;
    \mathbf{a}\in J_1,\ \mathbf{b}\in J_2,\ \mathbf{c}\in J_3
    \,\Big\}.
    \]
    Let $B = \{(\mathbf{a}, \mathbf{b}, \mathbf{c})\}$.
    Set $\mathbf{u}_1 = \mathbf{a}$, $\mathbf{u}_2 = \mathbf{b}$,
    $\mathbf{u}_3 = \mathbf{c}$, and
    $\mathbf{u} = (\mathbf{u}_1, \mathbf{u}_2, \mathbf{u}_3)$.
    We construct a corresponding decycling family
    $\mathcal{W} = (X_1, X_2, \ldots, X_s, Y_1, Y_2, \ldots, Y_k)$
    as follows.
    For $1\leq j\leq k$ and $1\leq i\leq 3$,
    the set $Y_j$ contains $u_i$ if and only if
    the $(s+j)$\textsuperscript{th} entry of $\mathbf{u}_i$ is $1$.
    By Lemma~\ref{lem:bad-triples-transitivity},
    $\mathcal{W}$ is a decycling family of $T$.
    Therefore, $\mathcal{Y}$ is a decycling family of $T_0$
    satisfying the weight constraints.

    Assume that the statement holds for all $n'$ with
    $3\leq n'\leq n-1$.
    Since $B\in \mathcal{B}_n$, there exists a set
    $B'\in \mathcal{B}_{n-1}$ and a vector $\mathbf{x}\in J_n$
    such that $B = S(B', \mathbf{x})$.
    As $B'\subseteq B$ and $B$ contains no bad triples,
    $B'$ also contains no bad triples.
    By the induction hypothesis, there exists a decycling family
    $\mathcal{W}' = (X_1, X_2, \ldots, X_s, Y_1', Y_2', \ldots, Y_k')$
    of $T - u_n$ respecting the weight constraints
    $(A_1, A_2, \ldots, A_{n-1})$.
    Let
    $\mathbf{u}' = (\mathbf{u}_1, \mathbf{u}_2, \ldots, \mathbf{u}_{n-1})$
    be the corresponding vector.
    Set
    $\mathbf{u} = (\mathbf{u}_1, \mathbf{u}_2, \ldots, \mathbf{u}_{n-1}, \mathbf{u}_n = \mathbf{x})$.
    By Lemma~\ref{lem:cwkit:transition-correct},
    $B = B(\mathbf{u})$.

    We construct a decycling family $\mathcal{W}$ of $T$ as follows.
    Let $\mathcal{W} = (X_1, X_2, \ldots, X_s, Y_1, Y_2, \ldots, Y_k)$.
    For each $1\leq j\leq k$, set
    $Y_j = Y_j' \cup \{u_n\}$ if the $(s+j)$\textsuperscript{th} entry of $\mathbf{u}_n$ is $1$,
    and set $Y_j = Y_j'$ otherwise.
    Clearly, $\mathbf{u}$ corresponds to $\mathcal{W}$, and
    $\mathcal{W}$ is a decycling family of $T$.
    Hence, $\mathcal{Y}$ is a decycling family of $T_0$
    respecting the weight constraints.
\end{proof}

Algorithm~\ref{alg:cwkit} is essentially the one that we described so far.

\begin{algorithm}[H]
\caption{An FPT algorithm for \textsc{CWKIT}}
\label{alg:cwkit}
\SetNlSty{textbf}{(}{)}
\DontPrintSemicolon
\KwIn{A tournament $T_0$ on vertices $u_1,\ldots,u_n$, an integer $k$, allowed-weight sets $\mathcal{A}=(A_1,\ldots,A_n)$, where $A_i$ is a nonempty
subset of $\{0,1,\ldots,k\}$, for $i\in [n]$, and a decycling family $\mathcal{X}$ for $T_0$}
\KwOut{\texttt{True} if $(T_0,k,\mathcal{A},\mathcal{X})$ is a yes-instance, \texttt{False} otherwise}

$J \gets \{0,1\}^{|\mathcal{X}|+k}$\;

\For{$i \gets 1$ \KwTo $n$}{
    $J_i \gets \{\, \mathbf{x}\in J \mid \mathrm{wt}(\mathbf{x}^{\mathrm{last}(k)})\in A_i \, \text{and $\mathbf{x}$ respects $\mathcal{X}$}\}$\;
}

\If{$n=1$ or $n=2$}{
    \Return{\texttt{True}}\;
}

$\mathcal{B}_3 \gets \Big\{\, \{(\mathbf{a},\mathbf{b},\mathbf{c})\} \;\Big|\; \mathbf{a}\in J_1,\ \mathbf{b}\in J_2,\ \mathbf{c}\in J_3 \,\Big\}$\;

\For{$t \gets 4$ \KwTo $n$}{
    $\mathcal{B}_t \gets \{\, S(B',\mathbf{x}) \mid B'\in \mathcal{B}_{t-1},\ \mathbf{x}\in J_t \,\}$\;
}

\ForEach{$B \in \mathcal{B}_n$}{
    \If{$B$ contains no bad triple}{
        \Return{\texttt{True}}\;
    }
}
\Return{\texttt{False}}\;
\end{algorithm}

\begin{lemma}
    \label{lem:cwkit-alg-correctness}
    Let $I = (T_0, k, \mathcal{A}, \mathcal{X})$ be an instance of \textsc{CWKIT}.
    Then Algorithm~\ref{alg:cwkit} returns \texttt{True} if $I$ is a yes-instance and returns \texttt{False} otherwise.
\end{lemma}

\begin{proof}
Assume that the algorithm returns \texttt{True}. Then the \Return statement in the test $n=1$ or $n=2$, or in the loop over $\mathcal{B}_n$, is executed.

If $n=1$ or $n=2$, then $T_0$ is a transitive tournament. Since each $A_i$ is assumed nonempty, we can choose weights in $A_i$ and define
$\mathcal{Y}=(Y_1,\ldots,Y_k)$ accordingly; thus $I$ is a yes-instance.

Otherwise, $n\geq 3$ and there exists a set $B\in \mathcal{B}_n$ such that $B$ contains no bad triples.
By Lemma~\ref{lem:cwkit:bt}, $\mathcal{B}_n$ is computed correctly by the update rule in the algorithm.
Then, by Lemma~\ref{lem:cwkit-yes-instance}, $I$ is a yes-instance.

For the other direction, assume that the algorithm returns \texttt{False}. This implies that $n\geq 3$ and there is no set $B\in \mathcal{B}_n$
such that $B$ contains no bad triples. By Lemma~\ref{lem:cwkit-yes-instance}, $I$ is a no-instance.
\end{proof}

Now Theorem~\ref{thm:wrkit} follows from Lemmas~\ref{lem:iterative-compression} and~\ref{lem:cwkit-alg-correctness}.
Note that the algorithm for \textsc{WKIT} involves one call to the algorithm for \textsc{KIT} and a call to the \textsc{CWKIT} algorithm.
The running time of both these algorithms is dominated by the size of $\mathcal{B}_n$, which is
$2^{2^{3(s+k)}}$, and in our application, $s=k$, hence $2^{2^{6k}}$.

\section{Block graphs}\label{sec:block}

In this section, we consider oriented input graphs for \textsc{$k$-Inversion} such that the underlying undirected graph is a block graph. 
We obtain an FPT algorithm using the algorithm we developed for WKIT and by performing dynamic programming over the block tree decomposition of the underlying graph.

Recall that a graph is a \textit{block graph} if each of its 2-connected components is a clique. Each such clique is known as a \textit{block}.

Let $D$ be an oriented graph and let $G$ be its underlying undirected graph. 
If $G$ is disconnected, then we can apply our algorithm separately on each connected component and combine the results straightforwardly. 
Therefore, we assume that $G$ is connected. 

Let $B_1, B_2, \ldots, B_p$ denote the blocks of $G$, and let $c_1, c_2, \ldots, c_q$ denote the cut vertices of $G$.
The block tree of $G$, denoted by $BT(G)$, is defined as follows. 
We introduce a set of $p$ vertices $B = \{b_1, b_2, \ldots, b_p\}$ corresponding to the blocks $B_1, B_2, \ldots, B_p$,
and a set of $q$ vertices $C = \{c_1, c_2, \ldots, c_q\}$ corresponding to the cut vertices of $G$ with the same labels.
The block tree $BT(G)$ has vertex set $U = B \cup C$. There is an edge $c_i b_j$ in $BT(G)$ if and only if $c_i \in B_j$.
Note that $BT(G)$ is a tree and that $B$ and $C$ are independent sets.

We root the tree $BT(G)$ at $b_p$.
The \textit{parent} and \textit{children} of a vertex in $BT(G)$ are defined
in the usual way. Let the vertex set $U$ of $BT(G)$ be $\{u_1, u_2, \ldots, u_{p+q}\}$ such that the parent of a vertex is numbered higher than the vertex.
We assume such an ordering for $B$ and $C$ as well. 
That is, a vertex in $B$ comes before its ancestors in $B$. Similarly for $C$.
Note that the choice of the root respects the above ordering. 

By $G_{u_i}$, we denote the subgraph of $G$ induced by all vertices in the blocks corresponding to the block vertices 
in the subtree rooted at $u_i$. By $D_{u_i}$, we denote the graph $G_{u_i}$ with edge orientations inherited from $D$.

For a cut vertex $c_i \in C$, we let $W(c_i)$ denote the set of all integers $w \in \{0,1,\ldots,k\}$ such that there exists a decycling family
$\mathcal{Y}$ of size $k$ of $D_{c_i}$ with $\mathrm{wt}(\mathbf{c}_i) = w$. 
Slightly differently, for a block vertex $b_i \in B \setminus \{b_p\}$ with $\mathrm{parent}(b_i) = c_j \in C$, 
we let $W(b_i)$ denote the set of all integers $w \in \{0,1,\ldots,k\}$ such that there exists a decycling family
$\mathcal{Y}$ of size $k$ of $D_{b_i}$ with $\mathrm{wt}(\mathbf{c}_j) = w$.
As a boundary case, for the root block vertex $b_p$, we let $W(b_p) = \{0,1,\ldots,k\}$ if $(D,k)$ is a yes-instance, and
$W(b_p) = \emptyset$ otherwise. 

We now show how to compute $W(u_i)$ in a bottom-up manner. 

\begin{lemma}
    \label{lem:block:wc}
    For a cut vertex $c_i \in C$, we have
    \[
        W(c_i) = \bigcap_{b_j \in \mathrm{children}(c_i)} W(b_j).
    \]
\end{lemma}
\begin{proof}
    Let $W' = \bigcap_{b_j \in \mathrm{children}(c_i)} W(b_j)$.
    Let $w \in W(c_i)$. Then there exists a decycling family $\mathcal{Y}$ of size $k$ of $D_{c_i}$ 
    such that $\mathrm{wt}(\mathbf{c}_i) = w$. 
    Recall that, for each $b_j \in \mathrm{children}(c_i)$, the graph $D_{b_j}$ is an induced subgraph of $D_{c_i}$ and contains $c_i$.
    Therefore, $\mathcal{Y}$ induces a vector of weight $w$ on $\mathrm{parent}(b_j) = c_i$, implying that $w \in W'$.

    For the converse direction, let $w \in W'$. Then, for every $b_j \in \mathrm{children}(c_i)$, there exists a decycling family
    $\mathcal{Y}_j = (Y_{j,1}, Y_{j,2}, \ldots, Y_{j,k})$ of $D_{b_j}$ that induces a vector of weight $w$ on
    $\mathrm{parent}(b_j) = c_i$.
    Without loss of generality, assume that $c_i \in Y_{j,\ell}$ for all $\ell \in \{1,2,\ldots,w\}$ and that
    $c_i \notin Y_{j,\ell}$ for all $\ell \in \{w+1,w+2,\ldots,k\}$.
    We construct a decycling family $\mathcal{Y} = (Y_1, Y_2, \ldots, Y_k)$ for $D_{c_i}$ by defining
    $Y_\ell = \bigcup_{b_j \in \mathrm{children}(c_i)} Y_{j,\ell}$ for each $\ell \in [k]$.
    Note that for each $\ell\in \{1,\ldots,w\}$, the sets $\{Y_{j,\ell}\}_{b_j \in \mathrm{children}(c_i)}$ pairwise intersect exactly at $c_i$,
    and for each $\ell\in \{w+1,w+2,\ldots,k\}$, the sets in $\{Y_{j,\ell}\}_{b_j \in \mathrm{children}(c_i)}$ are pairwise disjoint.
    Hence, the effect of applying $\mathcal{Y}$ on $D_{b_j}$ is identical to that of applying $\mathcal{Y}_j$.
    It follows that $\mathcal{Y}$ is a decycling family of $D_{c_i}$ inducing weight $w$ on $c_i$.
\end{proof}

\begin{lemma}
    \label{lem:block:wb}
    Let $b_i \in B \setminus \{b_p\}$. Then $W(b_i)$ is the set $W'$ of all $w \in \{0,1,\ldots,k\}$
    such that \textsc{WKIT}$(B_i, k, \mathcal{A})$ is a yes-instance, where $\mathcal{A}$
    consists of the sets $A(v)$ for $v \in B_i$, defined as follows:
    $A(v) = \{w\}$ if $v = \mathrm{parent}(b_i)$,
    $A(v) = W(c_j)$ if $v = c_j \in \mathrm{children}(b_i)$,
    and $A(v) = \{0,1,\ldots,k\}$ otherwise.
\end{lemma}
\begin{proof}
    Let $c_\ell = \mathrm{parent}(b_i)$.
    Let $w \in W(b_i)$. This implies that there exists a decycling family
    $\mathcal{Y} = (Y_1, Y_2, \ldots, Y_k)$ of $D_{b_i}$ such that it induces weight $w$ on $c_\ell$.
    Assume that it induces weight $w_j$ on $c_j \in \mathrm{children}(b_i)$.
    Clearly, $w_j \in W(c_j)$. Then \textsc{WKIT}$(B_i, k, \mathcal{A})$ is a yes-instance, where
    $\mathcal{A}$ is defined as in the lemma statement. Indeed, $\mathcal{Y}$ is a decycling family of $B_i$
    satisfying the weight constraints. 
    Therefore, $w \in W'$.

    For the other direction, assume that $w \in W'$. This implies that
    \textsc{WKIT}$(B_i, k, \mathcal{A})$ is a yes-instance, where $\mathcal{A}$ is as defined in the lemma statement. 
    Then there exists a decycling family $\mathcal{X} = (X_1, X_2, \ldots, X_k)$ of $B_i$
    satisfying the weight constraints, that is, $\mathcal{X}$ induces a weight in $A(v)$ for each vertex $v \in B_i$.
    We obtain a decycling family $\mathcal{Y} = (Y_1, Y_2, \ldots, Y_k)$ of 
    $D_{b_i}$ which induces weight $w$ on $c_\ell$ as follows.

    Let $w_j$ be the weight induced by $\mathcal{X}$ on $c_j \in \mathrm{children}(b_i)$.
    Since $A(c_j) = W(c_j)$, it follows that $w_j \in W(c_j)$. 
    Let $\mathcal{Y}_j = (Y_{j,1}, Y_{j,2}, \ldots, Y_{j,k})$ be a decycling family of $D_{c_j}$
    inducing weight $w_j$ on $c_j$, for each $c_j \in \mathrm{children}(b_i)$.
    Assume without loss of generality that $c_j \in Y_{j,a}$ if and only if $c_j \in X_a$
    (renumber the sets in $\mathcal{Y}_j$ accordingly).
    Now define
    \[
        Y_a = X_a \cup \bigcup_{c_j \in \mathrm{children}(b_i)} Y_{j,a}.
    \]
    Clearly, $\mathcal{Y}$ induces weight $w$ on $c_\ell$, since none of the sets in $\mathcal{Y}_j$
    (for $c_j \in \mathrm{children}(b_i)$) contains $c_\ell$.
    It remains to prove that $\mathcal{Y}$ is a decycling family of $D_{b_i}$.
    Note that no two vertices in $B_i$ come together in a set in $\mathcal{Y}_j$.
    Therefore, the sets of $\mathcal{X}$ added to the sets in $\mathcal{Y}$ are solely responsible for changing the 
    adjacency of edges inside $B_i$. Then the sets of $\mathcal{X}$ added to the sets in $\mathcal{Y}$ make sure that no cycle remains in $D_{b_i}\oplus \mathcal{Y}$.
    Note also that, for $c_j, c_{j'} \in \mathrm{children}(b_i)$ with $j \neq j'$,
    the graphs $D_{c_j}$ and $D_{c_{j'}}$ do not share any common vertex.
    Therefore, any set in $\mathcal{Y}_j$ is pairwise disjoint with any set in $\mathcal{Y}_{j'}$. 
    Moreover, the sets in $\mathcal{Y}_j$ added to the sets in $\mathcal{Y}$ kill the cycles in $G_{c_j}$.
    Therefore, $\mathcal{Y}$ is a decycling family of $G_{b_i}$.
\end{proof}

Thus we reach Algorithm~\ref{alg:block}. 

\begin{algorithm}[H]
\caption{An FPT algorithm for \textsc{$k$-Inversion} for oriented graphs with underlying block graphs}
\label{alg:block}
\SetNlSty{textbf}{(}{)}
\DontPrintSemicolon
\KwIn{An oriented graph $D$ where the underlying undirected graph $G$ is a connected block graph, and an integer $k$}
\KwOut{\texttt{True} if $(D, k)$ is a yes-instance, \texttt{False} otherwise}

Compute $BT(G)$ with $U = V(BT(G)) = \{u_1, u_2, \ldots, u_{p+q}\}$\;
Let $B = \{b_1, b_2, \ldots, b_p\}$ be the block vertices in $U$ corresponding to the blocks $B_1, B_2, \ldots, B_p$\;
Let $C = \{c_1, c_2, \ldots, c_q\}$ be the vertices in $U$ corresponding to the cut vertices in $G$ with the same labels\;

\For{$i \gets 1$ \KwTo $p+q$}{
    \eIf{$u_i = c_j$ for some $j\in [q]$}{
        $W(c_j)\gets \displaystyle\bigcap_{b_\ell \in \mathrm{children}(c_j)} W(b_\ell)$\;
    }{
        Let $u_i = b_j$ for some $j\in [p]$\;
        $W(b_j) \gets \emptyset$\;
        \ForEach{$w\in \{0, 1, \ldots, k\}$}{
            \ForEach{$v\in B_j$}{
                \eIf{$v = \mathrm{parent}(b_j)$}{
                    $A(v) \gets \{w\}$\;
                }{
                    \eIf{$v \in \mathrm{children}(b_j)$}{
                        $A(v)\gets W(v)$\;
                    }{
                        $A(v)\gets \{0, 1, \ldots, k\}$\;
                    }
                }
            }
            Let $\mathcal{A}$ be the tuple $(A(v))_{v\in B_j}$\;
            \If{\textsc{WKIT}$(B_j, k, \mathcal{A})$}{
                $W(b_j) \gets W(b_j)\cup \{w\}$\;
            }
        }
        \If{$i = p+q$}{
            \If{$W(b_j)\neq \emptyset$}{
                \Return{\texttt{True}}\;
            }
            \Return{\texttt{False}}\;
        }
    }
}
\end{algorithm}

\begin{lemma}
    \label{lem:block:alg}
    Algorithm~\ref{alg:block} returns \texttt{True} if the instance $(D,k)$ is a yes-instance, and returns \texttt{False} otherwise.
\end{lemma}
\begin{proof}
    Lemma~\ref{lem:block:wc} essentially states that the algorithm computes $W(c_i)$, for $c_i\in C$, correctly and Lemma~\ref{lem:block:wb} implies 
    that the algorithm computes $W(b_i)$, for $b_i\in B\setminus \{b_p\}$, correctly. 
    
    Assume that the algorithm returns \texttt{True}. This implies that the algorithm computes $W(b_p) \neq \emptyset$. Note that,
    in the algorithm, we compute $W(b_p)$ in exactly the same way as for the other block vertices. The only difference in execution
    arises from the fact that $b_p$ has no parent. Therefore, $A(v)$ is set to $\{0,1,\ldots,k\}$
    for every vertex other than the cut vertices in $B_p$. This implies that
    \textsc{WKIT}$(B_p, k, \mathcal{A})$ returns \texttt{True} for at least one choice of $w$, and hence for all choices of $w$.
    Indeed, the value of $w$ has no effect on the iteration for $b_p$.
    Let $\mathcal{X}$ be any decycling family for which \textsc{WKIT}$(B_p, k, \mathcal{A})$ returns \texttt{True}. By combining
    $\mathcal{X}$ with the decycling families of $D_{c_j}$ having matching weights on $c_j$ (as induced by $\mathcal{X}$),
    as done in the proof of Lemma~\ref{lem:block:wb}, we obtain a decycling family for $D = D_{b_p}$. Therefore, $(D,k)$
    is a yes-instance.

    Now assume that the algorithm returns \texttt{False}. This implies that, for no choice of $w$,
    \textsc{WKIT}$(B_p, k, \mathcal{A})$ returns \texttt{True}. Consequently, there is no decycling family for $B_p$
    such that the weight induced on each $c_j$ belongs to $W(c_j)$, for all $c_j \in \mathrm{children}(b_p)$.
    It follows that there is no decycling family for $D$.
\end{proof}

Now, Theorem~\ref{thm:block} follows from Lemma~\ref{lem:block:alg} and the fact that the running time is dominated by polynomially many executions of WKIT
algorithm.

\section{An FPT Algorithm for $k$-Inversion}
\label{sec:tw}

In this section, we obtain an FPT algorithm for
\textsc{$k$-Inversion} parameterized by $k+\mathrm{tw}$,
where $\mathrm{tw}$ denotes the treewidth of the underlying undirected graph of the input graph.
The algorithm is based on dynamic programming over a nice tree
decomposition. We remark that the problem can be expressed by an MSO$_2$ formula whose
length depends only on $k$. By Courcelle's
theorem~\cite{courcelle1990monadic}, this implies that \textsc{$k$-Inversion}, parameterized by $k+\mathrm{tw}$,
admits an FPT algorithm. Our algorithm is explicit and far more efficient than the complexity of the algorithm implied by Courcelle's theorem.

The key idea is simple.
To decide whether there exists a solution for the subgraph induced by
$V_t$ (the union of vertices appearing in the bags of the subtree rooted
at $t$), it suffices to know whether solutions exist for the child
nodes (or children, in the case of join nodes) that impose certain
reachability constraints on pairs of vertices in the bag $X_t$.

To capture this information, we maintain a table
$c(t,P,\mathcal{S})$, where $t$ is a node of the nice tree decomposition,
$P$ is a subset of ordered pairs of vertices from $X_t$,
and $\mathcal{S}$ is a $k$-tuple of subsets of $X_t$.
Intuitively, the table entry $c(t,P,\mathcal{S})$ is set to
\texttt{True} if and only if there exists a solution
$\hat{\mathcal{S}}$, which is a $k$-tuple of subsets of $V_t$,
such that the componentwise intersection of $\hat{\mathcal{S}}$ with
$X_t$ equals $\mathcal{S}$, and for every pair $(u,v)$ of vertices in
$X_t$, there exists a directed path from $u$ to $v$ in
$D[V_t]\oplus \hat{\mathcal{S}}$ if and only if $(u,v)\in P$.


Let $(T,\mathcal{X})$ be a nice tree decomposition of the underlying undirected graph of $D$. 
We denote the bags in $\mathcal{X}$ 
by $X_1, X_2, \ldots$. For a node $X_t \in V(T)$, let $V_t$ denote the union of the vertex sets in the bags corresponding to the nodes in the subtree of $T$ rooted at $X_t$. 

Throughout this section, $\mathcal{S}$ and $\hat{\mathcal{S}}$ denote $k$-tuples. 
The $i$\textsuperscript{th} set in $\mathcal{S}$ and $\hat{\mathcal{S}}$ is denoted by $S_i$ and $\hat{S}_i$, respectively. We write $\mathcal{S} \subseteq \hat{\mathcal{S}}$ (equivalently, $\hat{\mathcal{S}} \supseteq \mathcal{S}$) if $S_i \subseteq \hat{S}_i$ holds for all $1 \leq i \leq k$. Likewise, for a vertex set $X$, we define
$
\hat{\mathcal{S}} \cap X := (\hat{S_1}\cap X, \hat{S_2}\cap X, \ldots, \hat{S_k}\cap X).
$
Similarly, 
$
\hat{\mathcal{S}} \cup X := (\hat{S_1}\cup X, \hat{S_2}\cup X, \ldots, \hat{S_k}\cup X),
$
and 
$
\hat{\mathcal{S}} \setminus X := (\hat{S_1}\setminus X, \hat{S_2}\setminus X, \ldots, \hat{S_k}\setminus X).
$
By $\hat{\mathcal{S}}\cup \mathcal{S}$ we denote the $k$-tuple, $
 (\hat{{S}_1}\cup {S}_1, \hat{S_2}\cup S_2, \ldots, \hat{S_k}\cup S_k).
$
That is, the operation is applied componentwise to each set in the tuple.
The set of all $k$-tuples formed by the subsets of a set $Y$ is denoted by $\mathcal{P}(Y)^k$. 

Let $P_t$ be the set of all ordered pairs of distinct vertices in $X_t$. For every subset $P \subseteq P_t$, and for every $k$-tuple $\mathcal{S}$ of subsets of $X_t$, we define $c(t,P,\mathcal{S})$ as follows:
\[
c(t, P, \mathcal{S}) =
\begin{cases}
\texttt{True}, & 
\begin{aligned}[t]
&\text{if there exists a $k$-tuple } \hat{\mathcal{S}} \in \mathcal{P}(V(D))^{k} \text{ such that } 
\hat{\mathcal{S}} \supseteq \mathcal{S},\\ 
&\hat{\mathcal{S}} \cap X_t = \mathcal{S}, \text{ and }\ D[V_t] \oplus \hat{\mathcal{S}} \text{ is a DAG in which } (u,v) \in P 
\iff \\
&\text{there is a directed path from $u$ to $v$ in } D[V_t] \oplus \hat{\mathcal{S}}.
\end{aligned} \\
\texttt{False}, & \text{otherwise.}
\end{cases}
\]

If there exists a $k$-tuple $\hat{\mathcal{S}}$ satisfying the conditions for $c(t, P, \mathcal{S})$ to be \texttt{True}, we say that $\hat{\mathcal{S}}$ \emph{witnesses} $c(t, P, \mathcal{S}) = \texttt{True}$.

Before getting into the technical details, we recall the following observation.

\begin{observation}[folklore]
    \label{obs:def}
    The graph $D \oplus \mathcal{S}$ has an edge $(u,v)$ if and only if either of the following conditions is true. 
    \begin{enumerate}
        \item $(u,v)$ is an arc in $D$ and $\{u,v\}$ is contained in an even number of sets in $\mathcal{S}$, or
        \item $(v,u)$ is an arc in $D$ and $\{u,v\}$ is contained in an odd number of sets in $\mathcal{S}$.
    \end{enumerate}
\end{observation}

We now derive recurrence formulations for $c(t, P, \mathcal{S})$ corresponding to each type of node in the tree decomposition $T$.

\subsection{Leaf nodes}
Let $X_t$ be a leaf node in $T$.
Recall that leaf nodes are empty bags. Therefore, $c(t, \emptyset, \mathcal{S}) =$ \texttt{True}, for the $k$-tuple $\mathcal{S}$ of empty sets.

\subsection{Introduce nodes}




Let $t$ be an introduce node and let $t'$ be its child.
Let $X_t\setminus X_{t'} = \{w\}$.
Assume that there exists a $k$-tuple $\hat{\mathcal{S}}'$ of subsets of
$V_{t'}$ witnessing that
$c(t',P'',\mathcal{S}')=\texttt{True}$.

Since the vertex $w$ is introduced at node $t$, the presence of $w$ may
create additional reachability relations in the graph
$D[V_t]\oplus \hat{\mathcal{S}}'$.
Therefore, in order to determine the value of $c(t,P,\mathcal{S})$,
we consider, for each subset $P''\subseteq P$ that contains no ordered
pair involving $w$, the effect of introducing $w$.

For this purpose, we construct an auxiliary directed graph
$A_{t,P'',\mathcal{S}}$ on the vertex set $X_t$ that captures the
reachability relation induced by
$D[V_t]\oplus \hat{\mathcal{S}}$, where $\hat{\mathcal{S}}$ is a
$k$-tuple of subsets of $V_t$ such that
$\hat{\mathcal{S}}\cap X_t=\mathcal{S}$ componentwise,
the reachability relation on $X_{t'}$ in
$D[V_{t'}]\oplus \hat{\mathcal{S}}$ is exactly $P''$,
and the reachability relation on $X_t$ in
$D[V_t]\oplus \hat{\mathcal{S}}$ is exactly $P$.


Let $P'\subseteq P_{t'}$, and let $\mathcal{S}\in (\mathcal{P}(X_t))^k$.

We define an auxiliary directed graph $A_{t, P', \mathcal{S}}$ as follows.
The vertex set of $A_{t, P', \mathcal{S}}$ is $X_t$ and the arc set is:
\[
E(A_{t, P', \mathcal{S}}) =
P' 
\begin{aligned}[t]
&\cup\{(w,v)\mid v\in X_{t'},\ ((w,v)\in E(D[X_t]) \text{ and } \{w,v\} \text{ is contained} \\
&\phantom{\cup}\text{ in an even number of sets in } \mathcal{S})\ \text{or } ((v,w)\in E(D[X_t]) \text{ and } \\
&\phantom{\cup}\{w,v\} \text{ is contained} \text{ in an odd number of sets in } \mathcal{S})\}\\
&\cup\{(u,w)\mid u\in X_{t'},\ ((u,w)\in E(D[X_t]) \text{ and } \{u,w\} \text{ is contained} \\ 
&\phantom{\cup}\text{in an even number of sets in } \mathcal{S})\ \text{or } ((w,u)\in E(D[X_t]) \text{ and } \\
&\phantom{\cup}\{u,w\} \text{ is contained in an odd number of sets in } \mathcal{S})\}.
\end{aligned}
\]

Let $\hat{\mathcal{S}}\in \mathcal{P}(V_t)^k$ be such that $\hat{\mathcal{S}}\cap X_t = \mathcal{S}$. 
Let $\hat{\mathcal{S}'} = \hat{\mathcal{S}}\setminus \{w\}$ (equivalently, $\hat{\mathcal{S}} \cap V_{t'}$, since $V_t = V_{t'}\cup \{w\}$).

Let $\hat{D} = D[V_t]\oplus \hat{\mathcal{S}}$ and $\hat{D'} = D[V_{t'}]\oplus \hat{\mathcal{S}'}$.
Observe that $\hat{D'}$ is obtained from $\hat{D}$ by deleting vertex $w$.

Let $P'\subseteq P_{t'}$ be the set of pairs $(a,b)$ such that there is a directed path from $a$ to $b$ in $\hat{D'}$.

\begin{observation}
    \label{obs:auxiliary:intro:gen}
    Let $u,v\in X_{t'}$. The arc $(u,w)$ is present in $\hat{D}$ if and only if it is an arc of $A_{t,P',\mathcal{S}}$.
    Similarly, the arc $(w,v)$ is present in $\hat{D}$ if and only if it is an arc of $A_{t,P',\mathcal{S}}$. 
\end{observation}
\begin{proof}
    Let $A = A_{t,P',\mathcal{S}}$.
    For the forward direction, assume that $(u,w)$ is an arc of $\hat{D}$. Then by Observation~\ref{obs:def}, either of the following cases holds.
    \begin{enumerate}
        \item $(u,w)$ is an arc of $D[X_t]$ and $\{u,w\}$ is contained in an even number of sets in $\hat{\mathcal{S}}$, or
        \item $(w,u)$ is an arc of $D[X_t]$ and $\{u,w\}$ is contained in an odd number of sets in $\hat{\mathcal{S}}$.
    \end{enumerate}
    Since $u,w\in X_t$ and $\mathcal{S} = \hat{\mathcal{S}}\cap X_t$, the parity of the number of sets containing $\{u,w\}$ is the same in $\mathcal{S}$ and in $\hat{\mathcal{S}}$. Therefore, in either case, $(u,w)\in E(A)$.

    For the backward direction, assume that $(u,w)\in E(A)$. Since no pair in $P'$ contains $w$, the construction of $A$ implies that one of the two cases above holds. Then Observation~\ref{obs:def} implies that $(u,w)$ is an arc of $\hat{D}$. The statement for arcs of the form $(w,v)$ is proved analogously.
\end{proof}

\begin{observation}
    \label{obs:auxiliary:intro:gen:2}
    For every pair $(u,v)\in P_{t'}$, there is a directed path from $u$ to $v$ in $\hat{D'}$
    if and only if there is a directed path from $u$ to $v$ in $A_{t, P', \mathcal{S}} - w$.
\end{observation}
\begin{proof}
    Let $A = A_{t, P', \mathcal{S}}$.
    For the forward direction, assume that there is a directed path from $u$ to $v$ in $\hat{D'}$. Then by the definition of $P'$, we have $(u,v)\in P'$. Hence $(u,v)\in E(A)$, and since $w\notin\{u,v\}$, the arc $(u,v)$ is present in $A-w$, yielding a directed path from $u$ to $v$ in $A-w$.

    For the backward direction, assume that there is a directed path $u x_1 x_2\ldots x_q v$ from $u$ to $v$ in $A - w$. Every arc on this path has both endpoints in $X_{t'}$. By construction, the only arcs of $A$ with both endpoints in $X_{t'}$ are those in $P'$. Therefore,
    \[
        \{(u,x_1), (x_1,x_2),\ldots, (x_{q-1},x_q),(x_q,v)\}\subseteq P'.
    \]
    By the definition of $P'$, for each arc $(y,z)$ in this set there is a directed path from $y$ to $z$ in $\hat{D'}$. Concatenating these directed paths yields a directed path from $u$ to $v$ in $\hat{D'}$.
\end{proof}

\begin{lemma}
    \label{lem:auxiliary:introduce:gen}
    For every pair $(u,v)\in P_t$, there is a directed path from $u$ to $v$ in $\hat{D}$
    if and only if there is a directed path from $u$ to $v$ in $A_{t,P',\mathcal{S}}$.
\end{lemma}
\begin{proof}
    Let $A = A_{t,P',\mathcal{S}}$.

    \smallskip
    \noindent\emph{Forward direction.}
    Assume that there is a directed path from $u$ to $v$ in $\hat{D}$.
    Since all directed paths are simple, the vertex $w$ appears on the path at most once.

    \smallskip
    \noindent\textbf{Case 1:} $u,v\in X_{t'}$.

    If $(u,v)\in P'$, then by construction $(u,v)$ is an arc of $A$, and we are done.
    Now assume that $(u,v)\notin P'$. Then there is no directed path from $u$ to $v$ in $\hat{D'}$.
    Hence the (simple) directed path from $u$ to $v$ in $\hat{D}$ must pass through $w$ exactly once.

    Let $u',v'\in X_{t'}$ be such that $(u',w)$ and $(w,v')$ are arcs of $\hat{D}$,
    and there is a directed path from $u$ to $u'$ in $\hat{D'}$ and a directed path from $v'$ to $v$ in $\hat{D'}$
    (where $u'=u$ or $v'=v$ is allowed).
    By definition of $P'$, we have $u'=u$ or $(u,u')\in P'$, and $v'=v$ or $(v',v)\in P'$.
    By Observation~\ref{obs:auxiliary:intro:gen}, the arcs $(u',w)$ and $(w,v')$ belong to $A$.
    Hence $A$ contains a directed path from $u$ to $v$.

    \smallskip
    \noindent\textbf{Case 2:} $u=w$.

    Since the path is simple, it starts with an arc $(w,v')$ for some $v'\in X_{t'}$.
    Either $v'=v$ or there is a directed path from $v'$ to $v$ in $\hat{D'}$.
    By Observation~\ref{obs:auxiliary:intro:gen}, $(w,v')$ is an arc of $A$,
    and if $v'\neq v$, then $(v',v)\in P'$ and hence there is a directed path from $v'$ to $v$ in $A-w$.
    Thus there is a directed path from $u$ to $v$ in $A$.

    \smallskip
    \noindent\textbf{Case 3:} $v=w$.

    This case is symmetric to Case~2.

    \smallskip
    \noindent\emph{Backward direction.}
    Assume that there is a directed path from $u$ to $v$ in $A$.

    \smallskip
    \noindent\textbf{Case 1:} $u,v\in X_{t'}$.

    If $(u,v)\in P'$, then by definition of $P'$ there is a directed path from $u$ to $v$ in $\hat{D'}$,
    and hence also in $\hat{D}$.
    Now assume that $(u,v)\notin P'$. Then by Observation~\ref{obs:auxiliary:intro:gen:2},
    there is no directed path from $u$ to $v$ in $A-w$.
    Hence the directed path from $u$ to $v$ in $A$ must pass through $w$ exactly once.

    Let $u',v'\in X_{t'}$ be such that $(u',w)$ and $(w,v')$ are arcs of $A$,
    and there is a directed path from $u$ to $u'$ in $A-w$ and from $v'$ to $v$ in $A-w$
    (allowing $u'=u$ or $v'=v$).
    By Observation~\ref{obs:auxiliary:intro:gen:2}, the paths in $A-w$ correspond to directed paths in $\hat{D'}$,
    and by Observation~\ref{obs:auxiliary:intro:gen}, the arcs $(u',w)$ and $(w,v')$ are arcs of $\hat{D}$.
    Concatenating these yields a directed path from $u$ to $v$ in $\hat{D}$.

    \smallskip
    \noindent\textbf{Case 2:} $u=w$.

    Then the directed path in $A$ starts with an arc $(w,v')$ for some $v'\in X_{t'}$.
    By Observation~\ref{obs:auxiliary:intro:gen}, this arc belongs to $\hat{D}$.
    If $v'\neq v$, then by Observation~\ref{obs:auxiliary:intro:gen:2} there is a directed path from $v'$ to $v$ in $\hat{D'}$.
    Hence there is a directed path from $u$ to $v$ in $\hat{D}$.

    \smallskip
    \noindent\textbf{Case 3:} $v=w$.

    This case is symmetric to Case~2.
\end{proof}

Algorithm~\ref{alg:introduce} computes (and returns) the value of $c(t, P, \mathcal{S})$ from the values computed for $X_{t'}$.

\begin{algorithm}[H]
\caption{Compute $c(t, P, \mathcal{S})$ for an introduce node $X_t$}
\label{alg:introduce}
\SetNlSty{textbf}{(}{)}
\DontPrintSemicolon
\KwIn{$t, P, \mathcal{S}$}
\KwOut{$c(t, P, \mathcal{S})$}

$P^\star \gets \{\, (u,v) \in P \mid u \neq w \text{ and } v \neq w \,\}$\;
$\overline{P}\gets P_t\setminus P$\;
$\mathcal{S}'\gets \mathcal{S}\setminus \{w\}$\;

\For{each $P''\subseteq P^\star$}{
    \If{not $c(t', P'', \mathcal{S}')$}{
        skip (this choice of $P''$ and continue with the next iteration)
    }
    
    Create the auxiliary graph $A = A_{t, P'', \mathcal{S}}$\;
    \If{$A$ is not a DAG}{
        skip (this choice of $P''$ and continue with the next iteration)
    }
    \For{each pair $(u,v)\in P$}{
        \If{there is no directed path from $u$ to $v$ in $A$}{
            skip (this choice of $P''$ and continue with the next iteration of the outer loop)
        }
    }
    \For{each pair $(u,v)\in \overline{P}$}{
        \If{there is a directed path from $u$ to $v$ in $A$}{
            skip (this choice of $P''$ and continue with the next iteration of the outer loop)
        }
    }
    \Return{\texttt{True}}
}
\Return{\texttt{False}}
\end{algorithm}

\begin{lemma}
    \label{lem:introduce}
    If the correct value of $c(t', Q, \mathcal{Z})$ is available for every $Q\subseteq P_{t'}$ and every $\mathcal{Z}\in \mathcal{P}(X_{t'})^k$,
    then Algorithm~\ref{alg:introduce} correctly computes $c(t, P, \mathcal{S})$ for an introduce node $X_t$ in time $2^{|X_t|^2}\cdot |X_t|^{O(1)}$.
\end{lemma}
\begin{proof}
    We prove that $c(t,P,\mathcal{S})$ is \texttt{True} if and only if the algorithm returns \texttt{True}.

    \smallskip
    \noindent\emph{Forward direction.}
    Assume that $c(t,P,\mathcal{S})$ is \texttt{True}. Let $\hat{\mathcal{S}}$ be a $k$-tuple that witnesses this, and let $\hat{\mathcal{S}'}=\hat{\mathcal{S}}\setminus\{w\}$.
    Set $\hat{D}=D[V_t]\oplus \hat{\mathcal{S}}$ and $\hat{D'}=D[V_{t'}]\oplus \hat{\mathcal{S}'}$.
    Let $P''$ be the set of pairs $(u,v)\in P_{t'}$ such that there is a directed path from $u$ to $v$ in $\hat{D'}$.
    Then $P''\subseteq P^\star$ because $\hat{D'}$ is an induced subgraph of $\hat{D}$, and $P^\star$ is exactly $P$ restricted to pairs avoiding $w$.

    Since $\hat{\mathcal{S}'}$ witnesses $c(t',P'',\mathcal{S}')=\texttt{True}$, the first skip is not taken.
    Moreover, since $\hat{D}$ is a DAG, Lemma~\ref{lem:auxiliary:introduce:gen} implies that $A_{t,P'',\mathcal{S}}$ is a DAG, so the DAG-check skip is not taken.

    Finally, by Lemma~\ref{lem:auxiliary:introduce:gen}, for every $(u,v)\in P$ there is a directed path from $u$ to $v$ in $\hat{D}$
    if and only if there is a directed path from $u$ to $v$ in $A$, and similarly for $(u,v)\in \overline{P}$.
    Hence both path-consistency loops pass, and the algorithm returns \texttt{True}.

    \smallskip
        \noindent\emph{Backward direction.}
Assume that the algorithm returns \texttt{True}.
We prove that $c(t,P,\mathcal{S})$ is \texttt{True}.

Since the algorithm returns \texttt{True}, there exists a set $P''\subseteq P'$ such that none of the \texttt{skip} statements
is executed.
Since the skip in line~6 is not taken, we have $c(t',P'',\mathcal{S}')=\texttt{True}$.
Let $\hat{\mathcal{S}}'\in\mathcal{P}(V_{t'})^k$ be a $k$-tuple witnessing this fact, and let
$\hat{D}' = D[V_{t'}]\oplus \hat{\mathcal{S}}'$.
Then $\hat{D}'$ is a DAG.

Define $\hat{\mathcal{S}}=\hat{\mathcal{S}}'\cup\mathcal{S}$.
Since $\hat{\mathcal{S}}'\cap X_{t'}=\mathcal{S}'$, it follows that
$\hat{\mathcal{S}}\cap X_t=\mathcal{S}$, recalling that $X_t=X_{t'}\cup\{w\}$ and
$\mathcal{S}'=\mathcal{S}\setminus\{w\}$.
Let $\hat{D}=D[V_t]\oplus \hat{\mathcal{S}}$ and let $A=A_{t,P'',\mathcal{S}}$.

Since the skip in line~9 is not taken, the auxiliary graph $A$ is a DAG.
Moreover, $\hat{D}'=\hat{D}-w$ is a DAG.
Hence, if $\hat{D}$ contains a directed cycle, then every such cycle must contain the vertex $w$.

Assume for contradiction that $\hat{D}$ contains a directed cycle $C$.
Let $C$ be a shortest directed cycle in $\hat{D}$.
Then $w\in V(C)$, and there exist vertices $u,v\in X_{t'}$ such that $(v,w)$ and $(w,u)$ are arcs of $C$.
Let $Q$ be the directed subpath of $C$ from $u$ to $v$ that avoids $w$.
By minimality of $C$, $Q$ is a simple directed path entirely contained in $\hat{D}'$.
Thus, $(u,v)\in P''$ by definition of $P''$.

By Observation~\ref{obs:auxiliary:intro:gen}, the arcs $(v,w)$ and $(w,u)$ belong to $A$.
Since $(u,v)\in P''\subseteq E(A)$, the graph $A$ contains the directed triangle
$u\to v\to w\to u$, contradicting the fact that $A$ is a DAG.
Hence, $\hat{D}$ is a DAG.

Now, since the skip in line~12 is not taken, for every pair $(u,v)\in P$ there is a directed path
from $u$ to $v$ in $A$.
By Lemma~\ref{lem:auxiliary:introduce:gen}, this implies that there is a directed path
from $u$ to $v$ in $\hat{D}$.

Similarly, since the skip in line~15 is not taken, for every pair $(u,v)\in \overline{P}$
there is no directed path from $u$ to $v$ in $A$.
Again by Lemma~\ref{lem:auxiliary:introduce:gen}, there is no directed path
from $u$ to $v$ in $\hat{D}$.

Therefore, $\hat{\mathcal{S}}$ witnesses that $c(t,P,\mathcal{S})=\texttt{True}$.
The complexity follows from the fact that $P^*$ has at most $|X_t|^2$ pairs and everything inside the 
main loop can be done in $|X_t|^{O(1)}$-time.
This completes the proof.
\end{proof}

\subsection{Forget nodes}
Let $t$ be a forget node and let $t'$ be its child, with
$X_{t'}\setminus X_t=\{w\}$.
To witness that $c(t,P,\mathcal{S})=\texttt{True}$, it suffices to have a
solution $\hat{\mathcal{S}}$ for $D[V_t]$ which is also a solution for
$D[V_{t'}]$ (note that $D[V_t]=D[V_{t'}]$), such that the reachability
relation on $X_{t'}$ may contain additional ordered pairs involving $w$,
and the componentwise intersection of $\hat{\mathcal{S}}$ with $X_{t'}$
may include $w$ in some of the sets.


Algorithm~\ref{alg:forget} computes (and returns) the value of $c(t,P,\mathcal{S})$
from the values computed for $X_{t'}$.

\begin{algorithm}[H]
\caption{Compute $c(t, P, \mathcal{S})$ for a forget node $X_t$}
\label{alg:forget}
\SetNlSty{textbf}{(}{)}
\DontPrintSemicolon
\KwIn{$t, P, \mathcal{S}$}
\KwOut{$c(t, P, \mathcal{S})$}

$R \gets P_{t'}\setminus P_t$\tcp*{pairs in $P_{t'}$ that contain $w$}
\For{each subset $R'\subseteq R$}{
    $P' \gets P \cup R'$\;
    \For{each $\mathcal{W}\in (\{\{w\},\emptyset\})^k$}{
        \If{$c(t', P', \mathcal{S}\cup \mathcal{W})$}{
            \Return{\texttt{True}}\;
        }
    }
}
\Return{\texttt{False}}\;
\end{algorithm}

\begin{lemma}
\label{lem:forget}
If the correct value of $c(t',Q,\mathcal{Z})$ is available for every $Q\subseteq P_{t'}$
and every $\mathcal{Z}\in \mathcal{P}(X_{t'})^k$, then Algorithm~\ref{alg:forget}
correctly computes $c(t,P,\mathcal{S})$ for a forget node $X_t$ in time $O(2^{2|X_t|+k})$.
\end{lemma}

\begin{proof}
We prove that $c(t,P,\mathcal{S})$ is \texttt{True} if and only if
Algorithm~\ref{alg:forget} returns \texttt{True}.

\smallskip
\noindent{Forward direction.}
Assume that $c(t,P,\mathcal{S})$ is \texttt{True}.
Let $\hat{\mathcal{S}}\in \mathcal{P}(V_t)^k$ witness this, and let
\[
\hat{D} \;=\; D[V_t]\oplus \hat{\mathcal{S}}.
\]
Since $X_t$ is a forget node with child $X_{t'}$ and $X_{t'}\setminus X_t=\{w\}$,
we have $V_{t'} = V_t$. Hence,
\[
\hat{D}' \;:=\; D[V_{t'}]\oplus \hat{\mathcal{S}} \;=\; \hat{D}.
\]

Let $P^\star$ be the set of all pairs $(u,v)\in P_{t'}$ such that there is a directed path
from $u$ to $v$ in $\hat{D}'$.
By the definition of $c(t,P,\mathcal{S})$, for every pair $(u,v)\in P_t$,
\[
(u,v)\in P \iff \text{there is a directed path from $u$ to $v$ in $\hat{D}$}.
\]
Since $\hat{D}'=\hat{D}$, it follows that $P^\star\cap P_t = P$.

Let $R$ be the set computed in line~1 of Algorithm~\ref{alg:forget}, that is,
$R = P_{t'}\setminus P_t$, and define $R' := P^\star\setminus P$.
Then $R'\subseteq R$, and $P^\star = P\cup R' = P'$.

Next, let $\mathcal{S}^\star := \hat{\mathcal{S}}\cap X_{t'}$.
Since $X_{t'} = X_t\cup\{w\}$ and $\hat{\mathcal{S}}\cap X_t = \mathcal{S}$,
for each $i\in\{1,\ldots,k\}$ we have
$S^\star_i\in\{S_i,\; S_i\cup\{w\}\}$.
Hence there exists $\mathcal{W}\in(\{\{w\},\emptyset\})^k$ such that
$\mathcal{S}\cup \mathcal{W} = \mathcal{S}^\star$ (componentwise union).

Therefore, $\hat{\mathcal{S}}$ witnesses that
$c(t',P^\star,\mathcal{S}\cup\mathcal{W})$ is \texttt{True}.
In the iteration of Algorithm~\ref{alg:forget} corresponding to this choice of
$R'$ and $\mathcal{W}$, the algorithm returns \texttt{True}.

\smallskip
\noindent{Backward direction.}
Assume that Algorithm~\ref{alg:forget} returns \texttt{True}.
Then there exist $R'\subseteq R$ and $\mathcal{W}\in(\{\{w\},\emptyset\})^k$ such that
$c(t',P',\mathcal{S}\cup\mathcal{W})$ is \texttt{True}, where $P' = P\cup R'$.

Let $\hat{\mathcal{S}}\in \mathcal{P}(V_{t'})^k$ witness this, and let
\[
\hat{D}' \;=\; D[V_{t'}]\oplus \hat{\mathcal{S}}.
\]
Since $\hat{\mathcal{S}}$ witnesses $c(t',P',\mathcal{S}\cup\mathcal{W})$, we have:
(i) $\hat{D}'$ is a DAG, and
(ii) for every $(u,v)\in P_{t'}$,
\[
(u,v)\in P' \iff \text{there is a directed path from $u$ to $v$ in $\hat{D}'$},
\]
and moreover $\hat{\mathcal{S}}\cap X_{t'} = \mathcal{S}\cup\mathcal{W}$.
In particular, since $\mathcal{W}$ only possibly adds $w$ (componentwise),
we obtain $\hat{\mathcal{S}}\cap X_t = \mathcal{S}$.

Because $V_t = V_{t'}$, define $\hat{D} := D[V_t]\oplus \hat{\mathcal{S}}$.
Then $\hat{D}=\hat{D}'$, and hence $\hat{D}$ is a DAG.

Finally, restrict the equivalence above to pairs $(u,v)\in P_t$.
Since $P' = P\cup R'$ and every pair in $R'$ contains $w$, we have $P'\cap P_t = P$.
Therefore, for every $(u,v)\in P_t$,
\[
(u,v)\in P \iff \text{there is a directed path from $u$ to $v$ in $\hat{D}$}.
\]
Thus, $\hat{\mathcal{S}}$ witnesses that $c(t,P,\mathcal{S})$ is \texttt{True}.
The complexity follows from the fact that $R$ has at most $2|X_t|$ pairs and the inner loop
runs in $2^k$ times.
\end{proof}

\subsection{Join nodes}
Let $t$ be a join node, and let $t_1$ and $t_2$ be its children.
For a $k$-tuple $\hat{\mathcal{S}}$ and a pair of vertices
$u,v\in X_t$, any directed path from $u$ to $v$ in
$D[V_t]\oplus \hat{\mathcal{S}}$ can be decomposed into a sequence of
directed subpaths, where each subpath is entirely contained in either
the subgraph induced by $V_{t_1}$ or the subgraph induced by $V_{t_2}$.

To capture this behavior, for reachability relations $Q_1$ and $Q_2$
corresponding to the children $t_1$ and $t_2$, respectively, we compute
the transitive closure of $Q_1\cup Q_2$ and construct an auxiliary graph
$B_{t,Q_1,Q_2}$ whose arc set corresponds to this transitive closure.

Recall that $X_t = X_{t_1} = X_{t_2}$ and hence $P_t = P_{t_1} = P_{t_2}$.

Let $Q_1, Q_2 \subseteq P_t$, and let $\mathcal{S} \in \mathcal{P}(X_t)^k$.

Recall that the \emph{transitive closure} $R^+$ of a binary relation $R$ is the smallest
transitive relation containing $R$; that is, $(u,v),(v,w)\in R^+$ implies $(u,w)\in R^+$.

Let $Q = Q_1 \cup Q_2$, and let $Q^+$ denote the transitive closure of $Q$.

We define an auxiliary directed graph $B_{t,Q_1,Q_2}$ as follows.
The vertex set of $B_{t,Q_1,Q_2}$ is $X_t$, and its arc set is $Q^+$.

Let $\hat{\mathcal{S}}\in \mathcal{P}(V_t)^k$ be such that
$\hat{\mathcal{S}}\cap X_t = \mathcal{S}$.
Let $\hat{\mathcal{S}}_1 = \hat{\mathcal{S}}\cap V_{t_1}$ and
$\hat{\mathcal{S}}_2 = \hat{\mathcal{S}}\cap V_{t_2}$.

Define
\[
\hat{D} = D[V_t]\oplus \hat{\mathcal{S}}, \quad
\hat{D}_1 = D[V_{t_1}]\oplus \hat{\mathcal{S}}_1, \quad
\hat{D}_2 = D[V_{t_2}]\oplus \hat{\mathcal{S}}_2 .
\]

\begin{lemma}
\label{lem:auxiliary:join}
Assume that for every pair $(u,v)\in P_t$,
there is a directed path from $u$ to $v$ in $\hat{D}_1$
if and only if $(u,v)\in Q_1$, and
there is a directed path from $u$ to $v$ in $\hat{D}_2$
if and only if $(u,v)\in Q_2$.
Then, for every pair $(u,v)\in P_t$,
there is a directed path from $u$ to $v$ in $\hat{D}$
if and only if there is an arc from $u$ to $v$ in $B_{t, Q_1, Q_2}$.
\end{lemma}

\begin{proof}
Let $B = B_{t,Q_1,Q_2}$.

\smallskip
\noindent{Forward direction.}
Assume that there is a directed path from $u$ to $v$ in $\hat{D}$.
Since $V_t = V_{t_1} \cup V_{t_2}$ and $X_t$ separates the two subtrees,
this path can be decomposed into a sequence of subpaths whose endpoints lie in $X_t$,
such that each subpath is entirely contained in either $\hat{D}_1$ or $\hat{D}_2$.

Hence, there exist vertices $w_1,\dots,w_r\in X_t$ such that for each consecutive pair
$(x,y)\in \{(u,w_1),(w_1,w_2),\dots,(w_r,v)\}$,
there is a directed path from $x$ to $y$ in either $\hat{D}_1$ or $\hat{D}_2$.
By assumption, each such pair belongs to $Q_1\cup Q_2 = Q$.
Therefore, $(u,v)\in Q^+$, and thus $(u,v)$ is an arc of $B$.

\smallskip
\noindent{Backward direction.}
Assume that $(u,v)$ is an arc of $B$.
Then $(u,v)\in Q^+$, so there exist vertices $w_1,\dots,w_r\in X_t$
such that each pair in
\[
(u,w_1),(w_1,w_2),\dots,(w_r,v)
\]
belongs to $Q_1$ or $Q_2$.
By assumption, each such pair corresponds to a directed path
in $\hat{D}_1$ or $\hat{D}_2$, and hence in $\hat{D}$.
Concatenating these directed paths yields a directed path from $u$ to $v$ in $\hat{D}$.
\end{proof}

Algorithm~\ref{alg:join} computes the value of $c(t,P,\mathcal{S})$
from the values computed for $X_{t_1}$ and $X_{t_2}$.

\begin{algorithm}[H]
\caption{Compute $c(t,P,\mathcal{S})$ for a join node $X_t$}
\label{alg:join}
\SetNlSty{textbf}{(}{)}
\DontPrintSemicolon
\KwIn{$t, P, \mathcal{S}$}
\KwOut{$c(t,P,\mathcal{S})$}

$\overline{P} \gets P_t \setminus P$\;

\For{each $Q_1, Q_2 \subseteq P$}{
    \If{not $c(t_1,Q_1,\mathcal{S})$ or not $c(t_2,Q_2,\mathcal{S})$}{
        skip (this choice of $Q_1, Q_2$ and continue with the next iteration)
    }
    Construct $B = B_{t,Q_1,Q_2}$\;
    \If{$B$ is not a DAG}{
        skip (this choice of $Q_1, Q_2$ and continue with the next iteration)
    }
    \For{each $(u,v)\in P$}{
        \If{$(u,v)$ is not an arc of $B$}{
            skip (this choice of $Q_1, Q_2$ and continue with the next iteration of the main loop)
        }
    }
    \For{each $(u,v)\in \overline{P}$}{
        \If{$(u,v)$ is an arc of $B$}{
            skip (this choice of $Q_1, Q_2$ and continue with the next iteration of the main loop)
        }
    }
    \Return{\texttt{True}}
}
\Return{\texttt{False}}
\end{algorithm}

\begin{lemma}
\label{lem:join}
If correct values of $c(t_1,Y_1,\mathcal{Z})$ and $c(t_2,Y_2,\mathcal{Z})$
are available for all $Y_1,Y_2\subseteq P_t$ and all $\mathcal{Z}\in\mathcal{P}(X_t)^k$,
then Algorithm~\ref{alg:join} correctly computes $c(t,P,\mathcal{S})$
for a join node $X_t$ in time $2^{O(|X_t|^2)}\cdot |X_t|^{O(1)}$.
\end{lemma}

\begin{proof}

For the forward direction, assume $c(t,P,\mathcal{S})$ is \texttt{True}, witnessed by
$\hat{\mathcal{S}}\in\mathcal{P}(V_t)^k$.
Define $\hat{\mathcal{S}}_i = \hat{\mathcal{S}}\cap V_{t_i}$ for $i\in\{1,2\}$,
and let $\hat{D},\hat{D}_1,\hat{D}_2$ be as above.

Let $Q_i\subseteq P_t$ be the set of pairs $(u,v)$
such that there is a directed path from $u$ to $v$ in $\hat{D}_i$.
Then $\hat{\mathcal{S}}_i$ witnesses $c(t_i,Q_i,\mathcal{S})=\texttt{True}$,
so the first skip condition is not triggered.

Since $\hat{D}$ is a DAG, Lemma~\ref{lem:auxiliary:join} implies that $B$ is a DAG.
Moreover, for $(u,v)\in P$ there is a directed path in $\hat{D}$ and hence an arc in $B$,
while for $(u,v)\in \overline{P}$ there is no such path and hence no arc in $B$.
Therefore, the algorithm returns \texttt{True}.

For the backward direction, assume the algorithm returns \texttt{True} for some $Q_1,Q_2$.
Then $c(t_i,Q_i,\mathcal{S})$ is \texttt{True} for $i=1,2$,
witnessed by $\hat{\mathcal{S}}_i$.
Let $\hat{\mathcal{S}} = \hat{\mathcal{S}}_1 \cup \hat{\mathcal{S}}_2$.

Since $B$ is a DAG, Lemma~\ref{lem:auxiliary:join} implies that
$\hat{D}=D[V_t]\oplus\hat{\mathcal{S}}$ is a DAG.
The arc tests in the algorithm ensure that reachability in $\hat{D}$
matches exactly the pairs in $P$.
Thus, $\hat{\mathcal{S}}$ witnesses $c(t,P,\mathcal{S})=\texttt{True}$.

The running time bound follows from the fact that each of $Q_1$ and $Q_2$
ranges over at most $2^{O({|X_t|}^2)}$ possibilities, while all operations inside
the main loop can be carried out in time polynomial in the bag size $|X_t|$.
\end{proof}

\subsection{Putting it all together}
Now we are ready to prove Theorem~\ref{thm:tw}.


\begin{proof}[Proof of Theorem~\ref{thm:tw}]
    There exists an algorithm running in time
    $2^{O(\mathrm{tw})}\cdot n^{O(1)}$ that computes a tree decomposition
    of width at most $2\cdot \mathrm{tw}+1$~\cite{korhonen2023single}.
    Given such a tree decomposition, it is well known that it can be
    transformed in polynomial time into a nice tree decomposition of the
    same width and with $O(\mathrm{tw}\cdot n)$ nodes.

    As discussed earlier, for a leaf node there is only a single table
    entry to compute, which can be done in constant time.
    For every other node, we apply
    Algorithm~\ref{alg:introduce},
    Algorithm~\ref{alg:forget}, or
    Algorithm~\ref{alg:join},
    depending on the type of the node, in a bottom-up manner.

    For each node $t$, we compute
    $2^{O(tw^2)}\cdot 2^{O(k\cdot \mathrm{tw})}$ table entries:
    there are $2^{O(tw^2)}$ possibilities for $P$
    and $2^{O(k\cdot \mathrm{tw})}$ possibilities for $\mathcal{S}$.
    Filling a single entry takes
    $2^{|X_t|^2}\cdot |X_t|^{O(1)}$ time for an introduce node,
    $2^{2|X_t|+k}\cdot |X_t|^{O(1)}$ time for a forget node,
    and $2^{O(|X_t|^2)}\cdot |X_t|^{O(1)}$ time for a join node.
    Since $|X_t|\leq \mathrm{tw}+1$, all these bounds are subsumed by
    $2^{O(\mathrm{tw}^2 + k\cdot \mathrm{tw})}$.
    As the nice tree decomposition has $O(\mathrm{tw}\cdot n)$ nodes,
    the overall running time is
    $2^{O(\mathrm{tw}(k + \mathrm{tw}))}\cdot n^{O(1)}$.

    By Lemmas~\ref{lem:introduce}, \ref{lem:forget}, and \ref{lem:join},
    the input graph $D$ is $k$-invertible if and only if
    $c(r,\emptyset,\emptyset)=\texttt{True}$,
    where $X_r$ is the root bag of the nice tree decomposition.
\end{proof}

\medskip
\noindent
\textbf{Concluding remarks.}
We conclude by posing the following question. In Section~\ref{sec:block}, we obtained a fixed-parameter tractable algorithm for \textsc{$k$-Inversion} when the underlying undirected graph of the input digraph is a block graph. It is unclear how to adapt this approach even to the case where the underlying graph can be covered by two cliques with a large intersection. Does \textsc{$k$-Inversion}, parameterized by $k$, admit an FPT algorithm on graphs with this structure?

%
%
\bibliographystyle{splncs04}
\bibliography{main}
\end{document}